\documentclass[conference]{IEEEtran}
\ifCLASSINFOpdf
\else
\fi
\usepackage[dvipdfmx]{graphicx,xcolor}
\usepackage{bm}
\usepackage{amsmath}
\usepackage{amsthm}
\usepackage{ascmac}
\theoremstyle{definition}
\newtheorem{thm}{Theorem}

\newtheorem{cor}{Corollary}

\begin{document}
%
\title{Lower Bounds on the Number of Writing Operations
  by ILIFC with Inversion Cells}

\author{\IEEEauthorblockN{Akira Yamawaki, Hiroshi Kamabe
and Shan Lu}
\IEEEauthorblockA{Graduate School of Engineering\\
Gifu University\\
1--1, Yanagido, Gifu, 501--1193\\
Email: yamawaki@kmb.info.gifu-u.ac.jp, kamabe@ieee.org,
shan.lu.jp@ieee.org}}

%


\maketitle

\begin{abstract}
  Index-less Indexed Flash Code (ILIFC)
  is a coding scheme for flash memories,
  in which one bit of a data sequence is
  stored in a slice consisting of several cells but
  the index of the bit is stored implicitly.
  Although several modified ILIFC schemes have been proposed,
  in this research
  we consider an ILIFC with inversion cells(I-ILIFC).
  The I-ILIFC reduces the total number of cell level changes
  at each writing request.
  Computer simulation is used to show that
  the I-ILIFC improves
  the average performance of the ILIFC in many cases.
  This paper presents our derivation of 
  the lower bounds on the number of writing operations by
  I-ILIFC and shows that
  the worst-case performance of the I-ILIFC is
  better than
  that of
  the ILIFC
  if the code length is sufficiently large.
  Additionally,
  we consider the tight lower bounds
  thereon.
  The results show that
  the threshold of the code length
  that
  determines whether the I-ILIFC improves
  the worst-case performance of the ILIFC
  is smaller than that in the first lower bounds.
\end{abstract}


%
\IEEEpeerreviewmaketitle

\section{Introduction}
In flash memory,
data bits are stored in cells
in the form of charge levels.
One of the most notable characteristics of flash memory
is the asymmetricity of charging and discharging operations.
That is,
the charge level of the cell can be increased
in a cell-by-cell manner
but cannot be decreased
in this manner.

Instead,
discharging is achieved by way of 
a special operation
known as
block erasure,
which
discharges the cells
in a long block simultaneously.
The disadvantage of the block erasure operation
is that it
partially
destroys cells in the flash memory
and thus
increases the error probability.
This necessitates the use of
an error correcting code.
However,
the cells
invariably become highly
unreliable after
block erasure is executed a certain number of times.

This led to the proposal of a flash code
to
reduce the number
of block erasure
operations\cite{Jia1,Jia2}.
Mahdavifar et al.
addressed the problem by proposing
the index-less indexed flash code (ILIFC).
The ILIFC is designed in terms of the worst-case performance\cite{Mah}.
In the ILIFC,
both
the value of one bit of data
and the index of the bit
are
stored in one slice.
The ILIFC uses the cell state space very efficiently
even though the code rate is not optimal.

Several modified ILIFC schemes
capable of improving
the performance of the ILIFC
have since been proposed
\cite{Ohk,Kaj2,Wad}.
In this paper,
we consider an ILIFC with inversion cells (I-ILIFC)\cite{Ohk}.
The I-ILIFC reduces the total number of cell level changes
at each writing request
in order to increase the number of writing operations
between two consecutive block erasures.
Computer simulation was used to show that
the I-ILIFC improves the average performance of the ILIFC
in many cases\cite{Ohk,Yam}.

This work theoretically shows that
the worst-case performance of the I-ILIFC is
better than
that of the ILIFC.
Firstly,
we derive the lower bounds on the number of writing operations
by I-ILIFC
and
specify a
threshold
for
the code length that determines whether
the I-ILIFC improves the worst-case performance of the ILIFC.
The results
show that the I-ILIFC is
better than
the ILIFC
in the worst case
if the code length is sufficiently large.

Additionally,
we consider
unusual writing
in addition to the usual writing by the I-ILIFC
and
determine
the tight lower bounds on the number of writing operations.
Consequently,
we show that
the threshold is smaller than that in the first lower bounds.

 

\section{Index-less Indexed Flash Code(ILIFC)}

In this work,
it is assumed that
the level of electric charge in a cell of a NAND flash memory
is in the range $A_q = \{0,1,\cdots,q-1\}$.
A block of data bits of length $k$ is encoded and stored in
a block of cells of length $n$.
An ILIFC that satisfies these conditions is denoted by
ILIFC$(n,k,q)$.

In the ILIFC
the block of cells of length $n$ is divided into slices consisting of
$k$ cells and, therefore,
the number of slices in the block is $m = \lfloor n/k \rfloor$.
If $n$ is not a multiple of $k$,
the remaining cells in the block are unused.
Each slice represents one bit of the $k$ data bits.
Since $k$ slices are used to store $k$ data bits,
we require $m \geq k$, that is, $n \geq k^2$.


The state of $m$ slices is denoted by
$(\bm{x}_1 \mid \bm{x}_2 \mid \cdots \mid \bm{x}_m)$,
where $\bm{x}_j \in A_q^{k}$ for $1 \leq j \leq m$.
For a slice $\bm{x} = (x_1,x_2,\cdots,x_k)$,
we define $wt(\bm{x}) = \sum_{i=1}^{k} x_i$ and
$bv(\bm{x}) = wt(\bm{x}) \bmod 2$.
$wt(\bm{x})$ is termed the weight of the slice $\bm{x}$.
A slice $\bm{x} = (x_1,x_2,\cdots,x_k)$ is said to be full
and to be empty
if $x_1 = x_2 = \cdots = x_k = q-1$ and
if $x_1 = x_2 = \cdots = x_k = 0$, respectively.
The slice is said to be active if
it is neither full nor empty.

In the ILIFC,
the value of the $i$-th bit in the $k$ data bits and the index $i$ of the bit
are stored in a slice as follows (See \cite{Mah} for details). 
In the initial state,
it is assumed that
all slices are empty and
all data bits are $0$.

Suppose that the value of the $i$-th bit is changed.
If none of the slices represent
the $i$-th bit,
an empty slice is reserved for the bit
and then the level of the $i$-th cell in the slice is changed to $1$.
In the case that
no empty slices exist,
block erasure
is incurred.

On the other hand,
if there is a slice representing the $i$-th bit,
the weight of the slice is increased by $1$.
In the beginning,
the level of the $i$-th cell in the slice is increased.
If the level of the $i$-th cell is $q-1$,
the level of the $i'$-th cell is increased
where $i' = (i \bmod k) + 1$.
Similarly,
if the level of the $i'$-th cell is also $q-1$,
the level of the $i''$-th cell is increased
where $i'' = (i' \bmod k) + 1$.
This procedure enables 
the value of the bit,
which
is represented by $bv(\bm{x})$,
to be obtained for the active slice $\bm{x}$.
Additionally,
the index of the bit is represented by
the position of the first updated cell in $\bm{x}$.
This updating procedure is performed until the slice
is filled to capacity.

Note that any full slice cannot represent the index.
In the ILIFC the value of
a bit
without any corresponding slice
is considered
to be
$0$.
Therefore,
for the full slice $\bm{x}'$,
$wt(\bm{x}') = k(q-1)$ should be even.
Thus,
in this work
it is assumed that
$k$ or $q-1$ is even.

The state of slices
$(\bm{x}_1 \mid \bm{x}_2 \mid \cdots \mid \bm{x}_m)$
enables
the $k$ bit data $(s_1,s_2,\cdots,s_k)$
to be
obtained as follows.
For each $i$,
$s_i = bv(\bm{x}_j)$ if there is a slice
$\bm{x}_j$ representing the $i$-th bit;
otherwise,
$s_i = 0$.
The function
that
maps $(\bm{x}_1 \mid \bm{x}_2 \mid \cdots \mid \bm{x}_m)$
to $(s_1,s_2,\cdots,s_k)$ is denoted by
$\mathcal{D}_s(\bm{x}_1 \mid \bm{x}_2 \mid \cdots \mid \bm{x}_m)$.

Usually,
a single writing operation to
flash memory
involves
changing a
single data bit stored in the memory\cite{Mah,Jia2}.
However,
in this research,
one writing operation
entails
updating
all the
cells
such that
the resulting cells represent the new data.
If the new data is equal to the current data,
then we assume that no writing operation
has occurred.
The number of writing operations
that
can
occur
between two
consecutive block erasures is simply
referred to as
the number of writings.
The number of writings depends on the sequence of data to be stored.
The minimum number of writings is
termed
the worst-case number of writings.

Assume that
the state of $m$ slices
$(\bm{x}_1 \mid \bm{x}_2 \mid \cdots \mid \bm{x}_m)$
is changed into
$(\bm{x}_1' \mid \bm{x}_2' \mid \cdots \mid \bm{x}_m')$
by one writing, where
$\bm{x}_j, \bm{x}_j' \in A_q^k$ for $1 \leq j \leq m$.
Then
$\sum_{j=1}^m \left( wt(\bm{x}_j') - wt(\bm{x}_j) \right)$ is
termed
the total number of changes in the cell level.

\section{ILIFC with inversion cells}

Suppose that
a writing operation in which
the current data $\bm{v}$ is changed into
the new data $\bm{v}'$ is
conducted
by ILIFC.
If such a writing can be
achieved
without block erasure,
the total number of cell level changes is equal to
the Hamming distance between $\bm{v}$ and $\bm{v}'$.
An ILIFC with inversion cells (I-ILIFC)
was proposed
in order to reduce the total number of cell level changes\cite{Ohk}.
The I-ILIFC has
two storing modes,
a normal mode and an inverted mode,
information about which is contained in the inversion cells.

In this research
it is assumed that
$k$ bit data
are
stored
in
a block of $n$ $q$-ary cells including $r$ inversion cells.
Such an I-ILIFC is denoted by I-ILIFC$(n,k,q,r)$.
In the I-ILIFC$(n,k,q,r)$,
a block of $(n-r)$ cells except for the $r$ inversion cells is
divided into slices consisting of $k$ cells.
These cells,
which are divided into slices,
are termed
data cells.
Hence,
there are $m = \lfloor (n-r)/k \rfloor$ slices.
The restriction of the ILIFC scheme,
$m \geq k$,
determines that
$n \geq k^2 + r$ should hold.

For 
$\bm{w} = (w_1,w_2,\cdots,w_l) \in \{0,1\}^l$,
we define
$\overline{\bm{w}} = (\overline{w_1},\overline{w_2},\cdots,\overline{w_l})$,
where $\overline{w_i}$ is $1$ if $w_i = 0$, and
$0$ if $w_i = 1$.
For $\bm{w},\bm{w}' \in \{0,1\}^l$,
let $d_H(\bm{w},\bm{w}')$ be the Hamming distance between
$\bm{w}$ and $\bm{w}'$.

In the I-ILIFC,
the storing mode is represented by $r$ inversion cells.
We denote the state of these $r$ inversion cells by
$\bm{b} = (b_1,b_2,\cdots,b_r) \in A_q^r$.
We denote the state of the inversion cells and $m$ slices by
$\bm{c} = (\bm{b} \mid \bm{x}_1 \mid \bm{x}_2 \mid \cdots \mid \bm{x}_m)$,
where $\bm{x}_j \in A_q^{k}$ for $1 \leq j \leq m$.
Suppose that the data $\bm{v} \in \{0,1\}^{k}$ is stored in
the cell state $\bm{c}$.
If $bv(\bm{b}) = 0$,
the cell is in the normal mode and
$\mathcal{D}_s(\bm{x}_1 \mid \bm{x}_2 \mid \cdots \mid \bm{x}_m) =
\bm{v}$ is satisfied.
If $bv(\bm{b}) = 1$,
the cell is in the inverted mode and
$\mathcal{D}_s(\bm{x}_1 \mid \bm{x}_2 \mid \cdots \mid \bm{x}_m) =
\overline{\bm{v}}$ is satisfied.
If
there is an $i$
that
satisfies $b_i < q-1$,
the mode is changed by increasing
$b_i$ by $1$.

Assume that
the state $(\bm{b} \mid \bm{x}_1 \mid \bm{x}_2 \mid \cdots \mid \bm{x}_m)$
is changed into
$(\bm{b}' \mid \bm{x}_1' \mid \bm{x}_2' \mid \cdots \mid \bm{x}'_m)$
by one writing operation.
Then $\sum_{j=1}^m \left( wt(\bm{x}_j') - wt(\bm{x}_j) \right)$
is
termed
the sum of the data cell level changes
and $wt(\bm{b}') - wt(\bm{b})$ is
termed
the sum of the inversion cell level changes.
The sum of these two values is
referred to as
the total number of cell level changes.

In the I-ILIFC,
when a writing operation is
executed,
one of two modes is selected
such that
the total number of cell level changes is minimized.
The following theorem holds\cite{Ohk}.

\begin{thm}
\label{rewritingrules}
Suppose
a writing operation,
in
which the current data $\bm{v}$ are changed into
the new data $\bm{v}'$, is
carried out
by I-ILIFC$(n,k,q,r)$, where
$\bm{v},\bm{v}' \in \{0,1\}^{k}$.
In order to minimize the total number of cell level changes,
the storing mode is changed by
writing
if and only if $d_H(\bm{v},\bm{v}') > (k+1)/2$.
\end{thm}
\begin{proof}
We denote $d_H(\bm{v},\bm{v}')$ by $d$.
When the mode is changed by
writing,
the sum of the data cell level changes is
$d_H(\bm{v},\overline{\bm{v}'}) = d_H(\overline{\bm{v}},\bm{v}')
= k - d$
and
the sum of the inversion cell level changes is $1$.
Hence, the total number of cell level changes is $(k-d+1)$.
On the other hand, when the mode is not changed by
writing,
the total number of cell level changes is equal to
the sum of the data cell level changes,
$d_H(\bm{v},\bm{v}') = d_H(\overline{\bm{v}},\overline{\bm{v}'}) = d$.
Therefore, if $d > k-d+1$, that is, $d > (k+1)/2$,
the writing
operation
that changes the mode is selected
such that
the total number of cell level changes is minimized.
Additionally,
if $d \leq (k+1)/2$,
the above discussion shows that
a
writing
operation
that does not change the mode is selected.
\end{proof}

For the state of $r$ inversion cells $\bm{b} = (b_1,b_2,\cdots,b_r) \in A_q^r$,
the inversion cells are said to be
exhausted
if $b_1 = b_2 = \cdots = b_r = q-1$, that is, $wt(\bm{b}) = r(q-1)$.
Then
writing
that
does not change the storing mode
occurs
until the next block erasure takes place.

\section{Average performance of I-ILIFC}
In this section,
computer simulation is used to show that
the average number of writings by I-ILIFC$(n,k,q,r)$
is greater than that by ILIFC$(n,k,q)$ in many cases
when the length of inversion cells $r$
is optimized\cite{Ohk,Yam}.

The restriction of the ILIFC scheme
determines that
$n-r \geq k^2$ should be satisfied.
If $(n-r) \bmod k \not= 0$,
there are data cells
that
will never be used for the slice.
Therefore, we consider only values of $r$
that
satisfy
$n-r \geq k^2$ and $(n-r) \bmod k = 0$.

For example,
the average
number of writings by
I-ILIFC$(640,16,4,r)$ and
I-ILIFC$(192,8,8,r)$ for each $r$
are shown in Fig. \ref{fig1} and Fig. \ref{fig2},
respectively.
Note that
I-ILIFC$(n,k,q,0)$ is equivalent to ILIFC$(n,k,q)$.
In our simulations,
the average number of writings was calculated
after $10,000$ block erasures took place.
\begin{figure}[t]
  \begin{picture}(80,140)(0,0)
  \put(10,10) {\includegraphics[width=7.5cm]{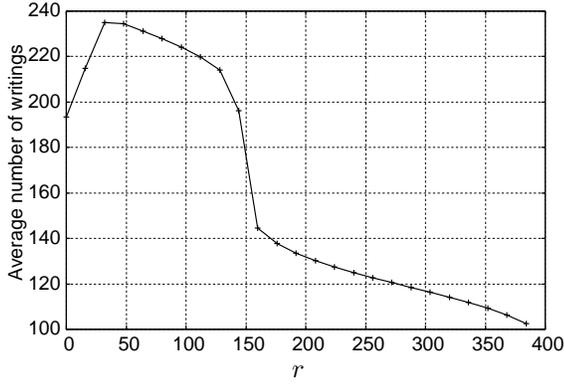}}
  \put(120,5){$r$}
  \end{picture}
  \caption{Average number of writings by I-ILIFC$(640,16,4,r)$}
  \label{fig1}
\end{figure}
\begin{figure}[t]
  \begin{picture}(80,140)(0,0)
  \put(10,10) {\includegraphics[width=7.5cm]{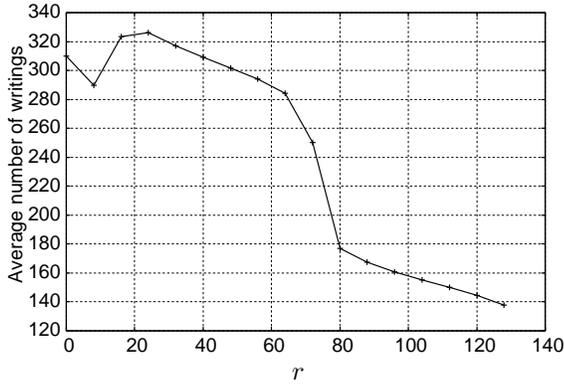}}
  \put(120,5){$r$}
  \end{picture}
  \caption{Average number of writings by I-ILIFC$(192,8,8,r)$}
  \label{fig2}
\end{figure}
It can be seen that
the average performance is maximized
at $r = 32$ for I-ILIFC $(640,16,4,r)$ and
at $r = 24$ for I-ILIFC$(192,8,8,r)$.
Similarly the average performance of
I-ILIFC$(288,12,4,r)$ is shown in Fig. \ref{fig3}.
\begin{figure}[t]
  \begin{picture}(80,140)(0,0)
  \put(10,10) {\includegraphics[width=7.5cm]{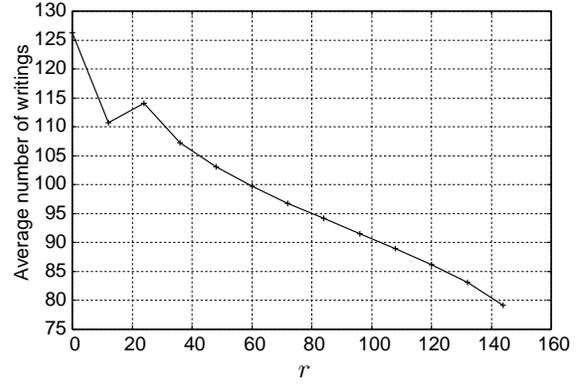}}
  \put(120,5){$r$}
  \end{picture}
  \caption{Average number of writings by I-ILIFC$(288,12,4,r)$}
  \label{fig3}
\end{figure}
This result
indicates that
the performance of I-ILIFC$(288,12,4,r)$
is maximized at $r = 0$.
That is,
the I-ILIFC does not improve the performance of the original ILIFC.

In this research
we analyze the worst-case performance
and
specify a
threshold that
determines whether the I-ILIFC improves
the worst-case performance of the ILIFC.
The results
show that the I-ILIFC is
better than
the ILIFC
if the code length is sufficiently large.

\section{Upper bound on the worst-case number of writings by ILIFC}

Under the definition of one writing operation in \cite{Mah},
it is shown that
the worst-case number of writings by
ILIFC$(n,k,q)$ is $k(\lfloor n/k \rfloor -k+1)(q-1) + k-1$\cite{Mah}.
Let $t_{w}$ be the worst-case number of writings by
ILIFC$(n,k,q)$ under the definition of one writing in this research.
In this section,
we derive the upper bound on $t_w$.

We denote $(0,0,\cdots,0)$ and $(1,1,\cdots,1)$ by $\bm{0}$ and $\bm{1}$,
respectively.
Let $T$ be the number of writings by ILIFC$(n,k,q)$ when
the data sequence is $\bm{1},\bm{0},\bm{1},\bm{0},\cdots$.
Then $t_{w} \leq T$ holds.
If the state of slices after such $T$ writings is
$(\bm{y}_1 \mid \bm{y}_2 \mid \cdots \mid \bm{y}_m)$,
$\sum_{j=1}^{m} wt(\bm{y}_j) = kT$ holds.
Additionally,
$\sum_{j=1}^{m} wt(\bm{y}_j) \leq m \cdot k(q-1) =
\lfloor n/k \rfloor \cdot k(q-1)
\leq n(q-1)$.
Consequently,
$kT \leq n(q-1)$ holds.
We denote the upper bound on $t_w$ by $t_{ub}$.
Then
from $t_{w} \leq T \leq n(q-1)/k$
we have
\begin{equation}
  t_{ub} = n(q-1)/k.
  \label{tub}
\end{equation}

\section{Maximum number of unused cell levels in I-ILIFC}

In this work,
it is assumed that a sufficient number of inversion cells are reserved
such that
the inversion cells are not
entirely consumed
whenever block erasure takes place.

When block erasure takes place, that is,
a writing operation
that
minimizes the total number of cell level changes without the erasure
is not possible,
we denote the state of slices by
$(\bm{y}_1 \mid \bm{y}_2 \mid \cdots \mid \bm{y}_m)$,
where $m = \lfloor (n-r)/k \rfloor$.
Then
the weight of each of the slices
$wt(\bm{y}_j)$ can be increased
$\left( k(q-1) - wt(\bm{y}_j) \right)$ more times.
The sum of such unused cell levels 
$\sum_{j=1}^{m} \left( k(q-1) - wt(\bm{y}_j) \right)$
is
termed
the number of unused cell levels.
In this section,
we
determine
the maximum number of unused cell levels.

For $\bm{v},\bm{v}' \in \{0,1\}^k (\bm{v} \not= \bm{v}')$,
let $d$ be the Hamming distance between $\bm{v}$ and $\bm{v}'$.
If the data $\bm{v}$ is changed into $\bm{v}'$
such that
the total number of cell level changes is minimized,
from Theorem \ref{rewritingrules}
the sum of the data cell level changes is $d$ if $d \leq (k+1)/2$
and $(k-d)$ if $d > (k+1)/2$.
Hence, if $k$ is even,
the maximum sum of the data cell level changes $\delta$ is as follows.
\begin{eqnarray}
  \delta & = &
  \max \{ \max_{1 \leq d \leq k/2} d, \max_{k/2+1 \leq d \leq k} (k-d) \} \nonumber \\
  & = & \max \{ k/2, k/2-1 \} = k/2. \nonumber
\end{eqnarray}
Similarly, if $k$ is odd,
\begin{eqnarray}
  \delta & = & \max \{ \max_{1 \leq d \leq (k+1)/2} d, \max_{(k+1)/2+1 \leq d \leq k} (k-d) \} \nonumber \\
  & = & \max \{ (k+1)/2, (k+1)/2-2 \} = (k+1)/2. \nonumber
\end{eqnarray}
Therefore,
\[ \delta = \begin{cases}
     k/2 & (k\mbox{\,\,is even}) \\
     (k+1)/2 & (k\mbox{\,\,is odd})
   \end{cases}.
\]

For the state of slices
$(\bm{y}_1 \mid \bm{y}_2 \mid \cdots \mid \bm{y}_m)$,
let $\alpha_1$ be the number of bits
without any
corresponding slice
and let $\alpha_2$ be the number of empty slices.
Then the next writing
that minimizes the total number of cell level changes
can always be
carried out
if and only if
the changes of any $\delta$ bits among $k$ data bits can be stored in slices,
that is,
\begin{eqnarray}
   (\alpha_1 < \delta \,\, \mbox{and} \,\, \alpha_2 \geq \alpha_1)
   \,\,\mbox{or}\,\,
   (\alpha_1 \geq \delta \,\, \mbox{and} \,\, \alpha_2 \geq \delta).
   \nonumber
\end{eqnarray}
This condition is equivalent to the following condition.
\begin{eqnarray}
   \alpha_2 \geq \min \{ \alpha_1, \delta \}.
   \nonumber
\end{eqnarray}
Therefore, block erasure may take place if and only if
\begin{eqnarray}
   \alpha_2 < \min \{ \alpha_1, \delta \}.
   \label{condition1}
\end{eqnarray}
The number of bits
that
have the corresponding slice
is $(k-\alpha_1)$.
Let $\bm{y}_{j_{i}}$ be
the slice
corresponding
to the $i$-th bit from the left
among
$(k-\alpha_1)$
such
bits.
Since $\bm{y}_{j_{i}}$ is active,
$1 \leq wt(\bm{y}_{j_{i}}) \leq k(q-1)-1$.
Then the number of unused cell levels is as follows.
\begin{equation}
   \sum_{i=1}^{k-\alpha_1} (k(q-1) - wt(\bm{y}_{j_i})) + \alpha_2 \cdot k(q-1).
   \label{unusedlevel1}
\end{equation}
When (\ref{condition1}) holds,
the maximum number of unused cell levels
is derived.
For fixed values of $\alpha_1$ and $\alpha_2$,
the number of unused cell levels is maximized
when $wt(\bm{y}_{j_1}) = \cdots = wt(\bm{y}_{j_{k-\alpha_1}}) = 1$.
Hence, the maximum is expressed as follows.
\begin{equation}
   (k - \alpha_1) (k(q-1) - 1) + \alpha_2 \cdot k(q-1).
   \nonumber
\end{equation}

\subsection*{When $\alpha_1 < \delta$}
From (\ref{condition1}), $\alpha_2 < \alpha_1$ holds.
For fixed $\alpha_1$,
when $\alpha_2 = \alpha_1 - 1$,
the maximum is expressed as follows.
\begin{eqnarray}
   & &
   (k - \alpha_1) (k(q-1) - 1) + (\alpha_1 - 1) \cdot k(q-1) \nonumber \\
   & = & (k-1) \cdot k(q-1) - k + \alpha_1.
   \nonumber
\end{eqnarray}
Therefore, when $\alpha_1 = \delta - 1$,
the maximum is as follows.
\begin{equation}
   (k-1) \cdot k(q-1) - k + \delta - 1.
   \label{unusedlevel4}
\end{equation}

\subsection*{When $\alpha_1 \geq \delta$}
From (\ref{condition1}), $\alpha_2 < \delta$ holds.
Similarly, for
a constant value of
$\alpha_1$, when $\alpha_2 = \delta - 1$,
the maximum is expressed as follows.
\begin{equation}
   (k - \alpha_1) (k(q-1) - 1) + (\delta - 1) \cdot k(q-1). \nonumber
\end{equation}
Hence, when $\alpha_1 = \delta$,
the maximum is as follows.
\begin{equation}
   (k-1) \cdot k(q-1) - k + \delta.
   \label{unusedlevel6}
\end{equation}
We have the following theorem.
\begin{thm}
\label{unusedlevelrule}
In the I-ILIFC$(n,k,q,r)$,
when block erasure takes place,
the number of unused cell levels $u$
satisfies the following inequality.
\[ u \leq (k-1) \cdot k(q-1) - k + \delta. \]
\end{thm}

\section{Lower bounds on the number of writings by I-ILIFC}

In this section,
we show the lower bounds on the worst-case number of writings
by I-ILIFC.

Let $(\bm{y}_1 \mid \bm{y}_2 \mid \cdots \mid \bm{y}_m)$
be the state of slices,
where $m = \lfloor (n-r)/k \rfloor$.
Then $\sum_{j=1}^{m} wt(\bm{y}_j)$
is
termed
the number of used cell levels.
From Theorem \ref{unusedlevelrule},
we have the following theorem.
\begin{thm}
\label{usedlevelrule}
In the I-ILIFC$(n,k,q,r)$,
when block erasure takes place,
the number of used cell levels $u'$ satisfies the following inequality.
\[ u' \geq (\lfloor (n-r)/k \rfloor-k+1) \cdot k(q-1) + k - \delta. \]
\end{thm}
\begin{proof}
Let $(\bm{y}_1 \mid \bm{y}_2 \mid \cdots \mid \bm{y}_m)$
be the state of slices.
From Theorem \ref{unusedlevelrule},
\[ \sum_{j=1}^{m} \left( k(q-1) - wt(\bm{y}_j) \right)
\leq (k-1) \cdot k(q-1) - k + \delta. \]
Therefore,
\begin{eqnarray}
   u' = \sum_{j=1}^{m} wt(\bm{y}_j)
   \geq (m-k+1) \cdot k(q-1) + k - \delta, \nonumber
\end{eqnarray}
where $m = \lfloor (n-r)/k \rfloor$.
\end{proof}
\begin{cor}
\label{rewritablerule}
If the number of used cell levels $u'$ satisfies the following inequality
\[ u' < (\lfloor (n-r)/k \rfloor-k+1) \cdot k(q-1) + k - \delta, \]
the next writing by I-ILIFC$(n,k,q,r)$ can always be
executed
without block erasure.
\end{cor}
\begin{proof}
This is the contraposition of Theorem \ref{usedlevelrule}.
\end{proof}

For fixed values of $n$, $k$ and $q$,
we define
\begin{eqnarray}
  U_1(r) & = & (\lfloor (n-r)/k \rfloor-k+1) \cdot k(q-1) + k - \delta,
  \nonumber \\
  U_1'(r) & = & ((n-r)/k -k+1) \cdot k(q-1) + k - \delta. \nonumber
\end{eqnarray}
Then we define $t_1(r) = \lceil U_1(r) / \delta \rceil$.
Let $r_1^{*}$ be the minimum integer $r$
that
satisfies
$r(q-1) \geq U_1'(r)/\delta + 1$.
That is,
$r_1^{*}$ is the integer $r$
that
satisfies $R_1 \leq r < R_1 + 1$,
where
\begin{equation}
   R_1 = \frac{n-k^2+k+k/(q-1)}{\delta + 1}. \nonumber
\end{equation}
The
restriction
on
the ILIFC scheme
requires
$n \geq k^2 + r_1^{*}$
to
be satisfied.
Note that
$n > k^2 + r_1^{*}$ is satisfied
if $n \geq k^2 + R_1 + 1$, that is,
\begin{equation}
   n \geq k^2 + \frac{k+1+k/(q-1)}{\delta} + 1
   \label{lengthcondition}
\end{equation}
holds.
Then from $R_1 > 0$,
$r_1^{*} \geq 1$ and $U_1(r_1^{*}) > 0$ hold.
In the following,
it is assumed that $(\ref{lengthcondition})$ is satisfied.

The following theorem holds.

\begin{thm}
\label{iilifclb1}
Let $t_1^{*}$ be the number of writings by I-ILIFC$(n,k,q,r_1^{*})$.
Then
\[ t_1^{*} \geq t_1(r_1^{*}). \]
\end{thm}

\begin{proof}
In the initial state
(i.e., immediately after block erasure)
the number of used cell levels is $0$ and $0 < U_1(r_1^{*})$
holds.
Hence,
Corollary \ref{rewritablerule} determines that
the first writing by I-ILIFC$(n,k,q,r_1^{*})$ can always be
executed.

For $t < t_1(r_1^{*})$
we suppose that $t$ writings by I-ILIFC$(n,k,q,r_1^{*})$ can
always
occur
without block erasure.
Let $(\bm{b} \mid \bm{y}_1 \mid \bm{y}_2 \mid \cdots \mid \bm{y}_m)$
be the state of inversion cells and slices after $t$ writings.
Then
$wt(\bm{b}) \leq t < t_1(r_1^{*}) < U_1(r_1^{*})/\delta + 1
\leq U_1'(r_1^{*})/\delta + 1
\leq r_1^{*}(q-1)$
because
when one writing operation
has occurred,
the maximum sum of the inversion cell level changes is $1$.
Hence, $r_1^{*}$ inversion cells are not used
in their entirety.
Additionally, $\sum_{j=1}^{m} wt(\bm{y}_j) \leq \delta t
< \delta \cdot U_1(r_1^{*})/\delta = U_1(r_1^{*})$.
Therefore,
according to
Corollary \ref{rewritablerule},
the next $(t+1)$-th writing can always
take place.

The above discussion
serves to confirm that
$t_1(r_1^{*})$ writings by I-ILIFC$(n,k,q,r_1^{*})$ can
always be
carried out
without erasure.
\end{proof}

We define
\begin{equation}
   U_{lb1}(r) = (n-k^2-r) (q-1) + k - \delta. \nonumber
\end{equation}
Then $t_1(r_1^{*}) = \lceil U_1(r_1^{*})/\delta \rceil \geq U_1(r_1^{*})/\delta
> U_{lb1}(r_1^{*})/\delta > U_{lb1}(R_1+1)/\delta$.
Let $t_{w1}^{*}$ be the worst-case number of writings by
I-ILIFC$(n,k,q,r_1^{*})$.
Since $t_{w1}^{*} \geq t_1(r_1^{*}) > U_{lb1}(R_1+1)/\delta$,
$U_{lb1}(R_1+1)/\delta$ is the lower bound on $t_{w1}^{*}$.

We calculate $t_{lb1}^{*} = U_{lb1}(R_1+1)/\delta$.
If $k$ is even,
\begin{equation}
   t_{lb1}^{*} = 2 \left( \frac{n-k^2-2}{k+2} - \frac{1}{k} \right) (q-1)
      + \frac{2k}{k+2} - 1.
   \label{tlbkeven}
\end{equation}
If $k$ is odd,
\begin{equation}
   t_{lb1}^{*}
   = 2 \left( \frac{n-k^2-3}{k+3} \right) (q-1)
      + \frac{2k}{k+3} - 1.
   \label{tlbkodd}
\end{equation}

We compare $t_{lb1}^{*}$ and $t_{ub}$,
where $t_{ub}$ is 
the upper bound on the worst-case number of writings by ILIFC$(n,k,q)$.
From (\ref{tub}),
$t_{ub} = n(q-1)/k$.
If $t_{ub} < t_{lb1}^{*}$ then
$t_{w} \leq t_{ub} < t_{lb1}^{*} < t_{w1}^{*}$,
that is,
the worst-case number of writings by
I-ILIFC$(n,k,q,r_1^{*})$ is greater than that by ILIFC$(n,k,q)$.
Therefore,
$t_{lb1}^{*} > t_{ub}$
is
the sufficient condition for improving the worst-case performance
of ILIFC$(n,k,q)$.
For $k \geq 4$,
it is supposed that
$t_{lb1}^{*} > t_{ub}$ is equivalent to
$n > p_1$.
Then $p_1$ is a threshold of the code length $n$
that determines whether I-ILIFC$(n,k,q,r_1^{*})$
improves the performance of ILIFC$(n,k,q)$
in the worst case.
In this paper,
$p_1$ is simply referred to as
the threshold.
From (\ref{tlbkeven}) and (\ref{tlbkodd}),
$p_1$ is derived as follows.
\begin{eqnarray}
  p_1 = \begin{cases}
    \frac{2(k^3+3k+2)}{k-2} - \frac{k}{q-1} & (k \mbox{ is even}) \\
    \frac{2k(k^2+3)}{k-3} - \frac{k}{q-1} & (k \mbox{ is odd})
  \end{cases}
  \label{lengthsufficiencycondition1}.
\end{eqnarray}
The results show that
the I-ILIFC$(n,k,q,r_1^{*})$
improves the worst-case performance of ILIFC$(n,k,q)$
if the code length $n$ is sufficiently large.

\section{Tight lower bound}

Thus far,
we assumed
that block erasure takes place when
a
writing operation
that
minimizes the total
number of
cell level changes
cannot be
accomplished
by I-ILIFC.
However,
at the moment it
remains
possible
to carry out
a
writing operation that
does not minimize
these changes.
In this paper, such a writing operation
is
termed
unusual.
The number of block erasures
can be reduced by ensuring that,
if
an
unusual writing
operation
can be accomplished without
erasure,
it occurs before erasure takes place.
Therefore, in this section,
it is assumed that block erasure takes place if neither
a
usual writing
operation
(which minimizes the total number of cell level changes)
nor
an
unusual writing
operation
can
occur.
Consequently,
we derive the tight lower bounds on the number of writings
under this assumption.

\subsection{Maximum number of unused cell levels}

We
determine
the maximum number of unused cell levels when
block erasure takes place.
We have the following theorem.

\begin{thm}
For the state of slices
$(\bm{y}_1 \mid \bm{y}_2 \mid \cdots \mid \bm{y}_m)$,
let $\beta_1$ be the number of bits
that
do not have
a
corresponding slice and
let
$\beta_2$ be the number of empty slices.
Then block erasure may take place if and only if
$\lfloor \beta_1/2 \rfloor > \beta_2$.
\end{thm}

\begin{proof}
Initially,
we show that,
if $\lfloor \beta_1/2 \rfloor \leq \beta_2$ holds,
either the next usual or
unusual writing can always be
executed
without block erasure.
It is supposed that
$l$ bits among $\beta_1$ bits
without any corresponding slice
are changed on the data (not on the sequence stored in the slices)
where $0 \leq l \leq \beta_1$.
\subsection*{When $l \leq \lfloor \beta_1/2 \rfloor$}
The inequality
$l \leq \beta_2$
determines that
the changes
that are made to
$l$ bits can be stored in $l$ slices
among $\beta_2$ empty slices.
Hence, writing
that
does not change the mode can be
carried out.
\subsection*{When $l > \lfloor \beta_1/2 \rfloor$}
If $\beta_1$ is even,
$\beta_1 - l < \beta_1 - \beta_1/2 = \lfloor \beta_1/2 \rfloor \leq \beta_2$.
If $\beta_1$ is odd,
from $\beta_1 - l < \beta_1 - (\beta_1-1)/2 = (\beta_1 - 1)/2 + 1 =
\lfloor \beta_1/2 \rfloor + 1$,
$\beta_1 - l \leq \lfloor \beta_1/2 \rfloor \leq \beta_2$.
Hence, the mode is changed
and the changes
related to
$(\beta_1-l)$ bits
can be stored in $(\beta_1-l)$ slices
among $\beta_2$ empty slices.
Therefore, writing that changes the mode can
occur.

The above discussion shows that
either
usual
or
unusual writing can be
executed
if $\lfloor \beta_1/2 \rfloor \leq \beta_2$.

Next,
we show that
if $\lfloor \beta_1/2 \rfloor > \beta_2$ holds,
block erasure may take place.
It is supposed that
$\lfloor \beta_1/2 \rfloor$ bits among $\beta_1$ bits
are changed on the data.
The inequality
$\lfloor \beta_1/2 \rfloor > \beta_2$
determines that
when writing that does not change the mode
occurs,
there is a bit for which an empty slice cannot be reserved.
On the other hand,
if $\beta_1$ is even,
$\beta_1 - \lfloor \beta_1/2 \rfloor = \beta_1/2 = \lfloor \beta_1/2 \rfloor
> \beta_2$.
If $\beta_1$ is odd,
$\beta_1 - \lfloor \beta_1/2 \rfloor = \beta_1 - (\beta_1 - 1)/2 =
(\beta_1 - 1)/2 + 1 = \lfloor \beta_1 / 2 \rfloor + 1 >
\lfloor \beta_1 / 2 \rfloor > \beta_2$.
Hence,
$\beta_1 - \lfloor \beta_1/2 \rfloor > \beta_2$
determines that
when writing that changes the mode
takes place,
there is a bit for which an empty slice cannot be reserved.

The above discussion indicates that
block erasure may take place
if $\lfloor \beta_1/2 \rfloor > \beta_2$.
\end{proof}

The number of unused cell levels is as follows
\begin{equation}
   \sum_{i=1}^{k-\beta_1} (k(q-1) - wt(\bm{y}_{j_i})) + \beta_2 \cdot k(q-1),
   \label{unusedlevel2}
\end{equation}
where $\bm{y}_{j_1},\cdots,\bm{y}_{j_{k-\beta_1}}$ are $(k-\beta_1)$ active slices,
that is, $1 \leq wt(\bm{y}_{j_i}) \leq k(q-1)-1$.
When $\lfloor \beta_1/2 \rfloor > \beta_2$ holds,
we derive the maximum number of unused cell levels.
For fixed values of $\beta_1$ and $\beta_2$,
the number of unused cell levels is maximized when
$wt(\bm{y}_{j_i}) = \cdots = wt(\bm{y}_{j_{k-\beta_1}}) = 1$.
Hence, the maximum is expressed as follows.
\begin{equation}
  (k-\beta_1) \left( k(q-1)-1 \right) + \beta_2 \cdot k(q-1). \nonumber
\end{equation}
\subsection*{When $\beta_1$ is even}
Then $\beta_1/2 > \beta_2$ holds.
For fixed $\beta_1$,
when $\beta_2 = \beta_1/2-1$,
the maximum is expressed as follows.
\begin{equation}
  \beta_1 \left( 1-k(q-1)/2 \right) + k \left( (k-1)(q-1) - 1 \right).
\end{equation}
Since $1-k(q-1)/2 \leq 0$ and
$\beta_1/2 - 1 \geq \beta_2 \geq 0$ hold,
when $\beta_2 = 0$ and $\beta_1 = 2$,
the maximum is as follows.
\begin{equation}
  (k-2) \cdot k(q-1) - k + 2. \label{unusedlevel2-1}
\end{equation}
\subsection*{When $\beta_1$ is odd}
Then $(\beta_1-1)/2 > \beta_2$ holds.
For a constant value of $\beta_1$,
when $\beta_2 = (\beta_1-1)/2 - 1$,
the maximum is expressed as follows.
\begin{equation}
  \beta_1 \left( 1 - k(q-1)/2 \right) + \left(k - 3/2 \right) \cdot k(q-1) - k.
  \nonumber
\end{equation}
Similarly, since $(\beta_1-1)/2-1 \geq \beta_2 \geq 0$ holds,
when $\beta_2 = 0$ and $\beta_1 = 3$,
the maximum is as follows.
\begin{equation}
  (k-3) \cdot k(q-1) -k+3. \label{unusedlevel2-2}
\end{equation}
From (\ref{unusedlevel2-1}) and (\ref{unusedlevel2-2}),
we have the following theorem.
\begin{thm}
\label{unusedlevelrule2}
In the I-ILIFC$(n,k,q,r)$,
when block erasure takes place,
the number of unused cell levels $u$
satisfies the following inequality.
\[ u \leq (k-2) \cdot k(q-1) -k+2. \]
\end{thm}
From Theorem \ref{unusedlevelrule2},
we have the following theorem.
\begin{thm}
\label{usedlevelrule2}
In the I-ILIFC$(n,k,q,r)$,
when block erasure takes place,
the number of used cell levels $u'$ satisfies the following inequality.
\[ u' \geq \left( \left\lfloor (n-r)/k \right\rfloor -k+2 \right) \cdot
  k(q-1) + k - 2. \]
\end{thm}
\begin{proof}
Let $(\bm{y}_1 \mid \bm{y}_2 \mid \cdots \mid \bm{y}_m)$
be the state of slices.
From Theorem \ref{unusedlevelrule2},
\[ \sum_{j=1}^{m} \left( k(q-1) - wt(\bm{y}_j) \right)
\leq (k-2) \cdot k(q-1) - k + 2. \]
Therefore,
\begin{eqnarray}
   u' = \sum_{j=1}^{m} wt(\bm{y}_j)
   \geq (m-k+2) \cdot k(q-1) + k - 2, \nonumber
\end{eqnarray}
where $m = \lfloor (n-r)/k \rfloor$.
\end{proof}
\begin{cor}
\label{rewritablerule2}
If the number of used cell levels $u'$ satisfies the following inequality
\[ u' < \left( \left\lfloor (n-r)/k \right\rfloor -k+2 \right) \cdot
  k(q-1) + k - 2, \]
either 
the next usual or unusual writing
by I-ILIFC$(n,k,q,r)$ can always be
executed
without block erasure.
\end{cor}
\begin{proof}
This is the contraposition of Theorem \ref{usedlevelrule2}.
\end{proof}

\subsection{Lower bounds on the worst case number of writings}
For fixed $n$, $k$, and $q$,
we define
\begin{eqnarray}
  U_2(r) & = &
  \left( \left\lfloor (n-r)/k \right\rfloor -k+2 \right) \cdot
  k(q-1) + k - 2, \nonumber \\
  U_2'(r) & = &
  \left( (n-r)/k -k+2 \right) \cdot
  k(q-1) + k - 2. \nonumber
\end{eqnarray}
Then we define $t_2(r) = \lceil (U_2(r)-U_1(r)-\delta+1)/(k-1) \rceil$.
Let $r_2^{*}$ be the minimum integer $r$
that
satisfies
$r(q-1) \geq U_1'(r)/\delta + (U_2'(r)-U_1'(r)-\delta+1)/(k-1) + 2$.
That is, $r_2^{*}$ is the integer $r$
that
satisfies
$R_2 \leq r < R_2 + 1$, where
\begin{eqnarray}
  & & R_2 \nonumber \\
  & = &
  \frac{1}{\delta + 1} \times \nonumber \\
  & & \left(
  n - k^2 + k +
  \frac{k+\delta}{q-1} + \frac{k\delta}{k-1} -
  \frac{\delta}{(q-1)(k-1)} \right). \nonumber
\end{eqnarray}
In order to satisfy $n \geq k^2 + r_2^{*}$,
it is assumed that $n \geq k^2 + R_2 + 1$, that is,
\begin{eqnarray}
  & & n \nonumber \\
  & \geq & k^2 + \frac{1}{\delta}
  \left( k + \frac{k+\delta}{q-1} + \frac{k\delta}{k-1}
  - \frac{\delta}{(q-1)(k-1)} \right) \nonumber \\
  & & + \frac{\delta+1}{\delta}
  \label{lengthcondition2}
\end{eqnarray}
is satisfied.
From (\ref{lengthcondition2}),
we have
$r_2^{*} \geq 1$ and
$U_1(r_2^{*}) > 0$.
The following theorem holds.

\begin{thm}
Let $t_2^{*}$ be the number of writings by I-ILIFC$(n,k,q,r_2^{*})$.
Then
\[ t_2^{*} \geq t_1(r_2^{*}) + t_2(r_2^{*}). \]
\end{thm}
\begin{proof}
In the initial state,
the number of used cell levels is $0$ and $0 < U_1(r_2^{*})$
holds.
Hence,
from Corollary \ref{rewritablerule},
the first usual writing by I-ILIFC$(n,k,q,r_2^{*})$ can always be
carried out.

For $t < t_1(r_2^{*})$
we suppose that $t$ usual writings by I-ILIFC$(n,k,q,r_2^{*})$ can
always be
executed
without block erasure.
Let $(\bm{b} \mid \bm{y}_1 \mid \bm{y}_2 \mid \cdots \mid \bm{y}_m)$
be the state of inversion cells and slices after $t$ writings.
Then 
since $wt(\bm{b}) \leq t < t_1(r_2^{*}) < t_1(r_2^{*}) + t_2(r_2^{*})
< U_1(r_2^{*})/\delta +
(U_2(r_2^{*})-U_1(r_2^{*})-\delta+1)/(k-1) + 2
\leq U_1'(r_2^{*})/\delta +
(U_2'(r_2^{*})-U_1'(r_2^{*})-\delta+1)/(k-1) + 2
\leq r_2^{*}(q-1)$,
$r_2^{*}$ inversion cells are not used
in their entirety.
Additionally,
$\sum_{j=1}^{m} wt(\bm{y}_j) \leq \delta t < \delta \cdot U_1(r_2^{*})/\delta
= U_1(r_2^{*})$.
Therefore,
according to
Corollary \ref{rewritablerule},
the next $(t+1)$-th usual writing can always
take place.

The
above discussion indicates that
$t_1(r_2^{*})$ usual writings can always be
executed
without block erasure.

Let $(\bm{b}' \mid \bm{y}_1' \mid \bm{y}_2' \mid \cdots \mid \bm{y}_m')$
be the state of inversion cells and slices after $t_1(r_2^{*})$
usual writings.
Then
$\sum_{j=1}^{m} wt(\bm{y}'_j) \leq \delta \cdot t_1(r_2^{*}) <
  \delta \cdot \left( U_1(r_2^{*})/\delta + 1 \right)
  = U_1(r_2^{*}) + \delta$.
Hence, $\sum_{j=1}^{m} wt(\bm{y}'_j) \leq U_1(r_2^{*}) + \delta - 1$.

For $t_1(r_2^{*}) \leq t' < t_1(r_2^{*}) + t_2(r_2^{*})$,
we suppose that
$t'$ writings by I-ILIFC$(n,k,q,r_2^{*})$
can always be
accomplished
without block erasure.
Note that each of
$t_1(r_2^{*})+1,t_1(r_2^{*})+2,\cdots,t'$-th writings
denotes
either
a
usual or unusual writing
operation.
Let $(\bm{b}'' \mid \bm{y}_1'' \mid \bm{y}_2'' \mid \cdots \mid \bm{y}_m'')$
be the state of inversion cells and slices after $t'$
writings.
Then, since $wt(\bm{b}'') \leq t' < t_1(r_2^{*}) + t_2(r_2^{*}) < r_2^{*}(q-1)$,
$r_2^{*}$ inversion cells are not
entirely used.
When either
usual or unusual writing
takes place,
the maximum sum of the data cell level changes is $(k-1)$.
Note that
an
unusual writing
operation,
for which the sum of the data cell level changes is $k$,
is not
executed
because such writing can occur by only changing the mode.
Hence,
$\sum_{j=1}^{m} wt(\bm{y}''_j) \leq
\sum_{j=1}^{m} wt(\bm{y}'_j) + (k-1) \cdot (t' - t_1(r_2^{*}))
\leq U_1(r_2^{*}) + \delta - 1 + (k-1) \cdot (t_2(r_2^{*})-1)
< U_1(r_2^{*}) + \delta - 1 + (k-1) \cdot
(U_2(r_2^{*})-U_1(r_2^{*})-\delta+1)/(k-1) = U_2(r_2^{*})$.
Therefore,
Corollary \ref{rewritablerule2} determines that
the next $(t'+1)$-th writing,
that is,
the next usual or unusual writing can always
take place.

Therefore, $\left( t_1(r_2^{*})+t_2(r_2^{*}) \right)$ writings
by I-ILIFC$(n,k,q,r_2^{*})$
can always
occur
without block erasure.
\end{proof}

We define
\begin{equation}
  U_{lb2}(r) = (n-r-k^2+k)(q-1) + k - 2. \nonumber
\end{equation}
Then $t_1(r_2^{*}) + t_2(r_2^{*}) =
\lceil U_1(r_2^{*})/\delta \rceil +
\lceil (U_2(r_2^{*})-U_1(r_2^{*}) - \delta + 1)/(k-1) \rceil
\geq U_1(r_2^{*})/\delta + (U_2(r_2^{*})-U_1(r_2^{*}) - \delta + 1)/(k-1)
> U_{lb1}(r_2^{*})/\delta + (U_{lb2}(r_2^{*})-U_{lb1}(r_2^{*}) - \delta + 1)/(k-1)
= U_{lb1}(r_2^{*})/\delta + (k(q-1)-1)/(k-1)
> U_{lb1}(R_2+1)/\delta + (k(q-1)-1)/(k-1)$.
We denote $U_{lb1}(R_2+1)/\delta + (k(q-1)-1)/(k-1)$
by $t_{lb2}^{*}$.
Let $t_{w2}^{*}$ be the worst-case number of writings by
I-ILIFC$(n,k,q,r_2^{*})$.
Since $t_{w2}^{*} \geq t_1(r_2^{*}) + t_2(r_2^{*}) > t_{lb2}^{*}$,
$t_{lb2}^{*}$ is the lower bound on $t_{w2}^{*}$.
$t_{lb2}^{*}$ is as follows. If $k$ is even,
\begin{eqnarray}
  & & t_{lb2}^{*} \nonumber \\
  & = & \frac{2}{k+2} \left( n - k^2 + \frac{k^3-6k^2+2k+4}{2k(k-1)} \right) (q-1) \nonumber \\
  & & + \frac{k^2-6k+4}{(k-1)(k+2)}. \label{tlbkeven2}
\end{eqnarray}
If $k$ is odd,
\begin{eqnarray}
  & & t_{lb2}^{*} \nonumber \\
  & = & \frac{2}{k+3} \left( n - k^2 + \frac{k^3-4k^2+k+6}{2(k+1)(k-1)} \right) (q-1) \nonumber \\
  & & + \frac{k^2-7k+4}{(k+3)(k-1)}. \label{tlbkodd2}
\end{eqnarray}

If $t_{lb2}^{*} > t_{ub}$,
the worst-case performance of I-ILIFC$(n,k,q,r_{2}^{*})$ is
better than
that of ILIFC$(n,k,q)$.
For $k \geq 4$,
it is supposed that
$t_{lb2}^{*} > t_{ub}$ is equivalent to
$n > p_2$.
From (\ref{tlbkeven2}) and (\ref{tlbkodd2}),
the threshold $p_2$ is as follows.
\begin{eqnarray}
  & & p_2 \nonumber \\
  & = & \begin{cases}
    \frac{2k^4-3k^3+6k^2-2k-4}{(k-1)(k-2)}
    - \frac{k(k^2-6k+4)}{(k-1)(k-2)(q-1)} & (k \mbox{ is even}) \\
    \frac{k(2k^4-k^3+2k^2-k-6)}{(k+1)(k-1)(k-3)}
    - \frac{k(k^2-7k+4)}{(k-1)(k-3)(q-1)} & (k \mbox{ is odd})
  \end{cases}.
  \nonumber \\
  \label{lengthsufficiencycondition2}
\end{eqnarray}

From (\ref{lengthsufficiencycondition1}) and
(\ref{lengthsufficiencycondition2}),
\begin{eqnarray}
  & & p_1 - p_2 \nonumber \\
  & = & \begin{cases}
    \frac{k^3(q-1)-3k^2+2k}{(k-1)(k-2)(q-1)} & (k \mbox{ is even}) \\
    \frac{(k^4+2k^3+k^2)(q-1)-3k^3-2k^2+k}{(k+1)(k-1)(k-3)(q-1)} &
    (k \mbox{ is odd})
  \end{cases}. \nonumber
\end{eqnarray}
Therefore, for $k \geq 4$,
we have $p_1 - p_2 > 0$, that is, $p_1 > p_2$.
This result shows that
I-ILIFC$(n,k,q,r_2^{*})$ improves
the worst-case performance of ILIFC$(n,k,q)$
also for $p_2 < n \leq p_1$.

\section{Conclusion}
This paper has presented our derivation of
the lower bounds on the number of writings
by the I-ILIFC
and
specified
the threshold
for
the code length which
determines whether the I-ILIFC improves
the worst-case performance of the ILIFC.
The results have
shown that
the I-ILIFC is
better than
the ILIFC in the worst case
if the code length is sufficiently large.
Additionally,
we have considered
unusual writing
operations
in addition to
the usual writing
operation
by the I-ILIFC
and derived the tight lower bounds
thereon.
Consequently,
the threshold could be made smaller than that in
the first lower bounds.




\section*{Acknowledgment}
This work was partially supported by
the Advanced Storage Research Consortium
and JSPS KAKENHI Grant Number
15K00010.



%

\end{document}